\documentclass[11pt]{article}%
\usepackage{amsfonts}
\usepackage{amssymb}
\usepackage{graphicx}
\usepackage{setspace}
\usepackage[round]{natbib}
\usepackage{amsmath}%
\usepackage{amsthm}%
\usepackage{relsize}

\usepackage{palatino}
\usepackage{paralist}

\usepackage[top=8pc,bottom=8pc,left=8pc,right=8pc]{geometry}

\newtheorem{theorem}{Theorem}

\newtheorem{corollary}[theorem]{Corollary}

\newtheorem{definition}[theorem]{Definition}

\newtheorem{lemma}[theorem]{Lemma}

  \newtheorem{assumption}{Assumption}

\newcommand{\point}{.}

\title{Nonparametric density estimation \\ by histogram trend filtering}
\author{Oscar Hernan Madrid Padilla \\
\emph{University of Texas at Austin}\\
\textit{oscar.madrid@utexas.edu}\\
\\
James Scott\\
\emph{University of Texas at Austin}\\
\textit{james.scott@mccombs.utexas.edu}\\
\\
}

\date{This version: \today}

\begin{document}

\maketitle

\begin{abstract}
We propose a novel approach for density estimation called histogram trend filtering.  Our estimator arises from looking at surrogate Poisson model for counts of observations in a partition of the support of the data. We begin by showing consistency for a variational estimator for this density estimation problem. We then study a discrete estimator that can be efficiently found via convex optimization.  We show that the estimator enjoys strong statistical guarantees, yet is much more practical and computationally efficient than other estimators that enjoy similar guarantees. Finally, in our simulation study the proposed method showed smaller averaged mean square error than competing methods.  This favorable blend of properties makes histogram trend filtering an ideal candidate for use in routine data-analysis applications that call for a quick, efficient, accurate density estimate.

\bigskip

\end{abstract}

\begin{spacing}{1}

\section{Introduction}

\subsection{Nonparametric density estimation}

Consider the classic problem of one-dimensional density estimation, where we observe $y_i \sim f$ for $i=1, \ldots, n$ and wish to estimate $f$.  Most data-analysis practitioners that confront this problem turn to kernel density estimation, due to its familiarity, its computational efficiency, and its well-understood statistical properties.  Yet kernel methods are known to suffer from the local-adaptivity problem, wherein the used of a fixed bandwidth parameter may result in simultaneously undersmoothing and oversmoothing in different regions of the density.

A huge variety of methods have been proposed that improve upon basic kernel methods in a way that addresses this problem, from adaptive kernel bandwidths to penalized-likelihood estimation.  Yet these methods typically either incur a much higher computational burden than basic kernel methods, or else they involve hyperparameters that are difficult to specify and tune.  The goal of this paper is address this gap.  We propose a method called histogram trend filtering, which solves the adaptivity problem while simultaneously satisfying all three of the following criteria:
\begin{enumerate}
\item It is computationally efficient, even for large data sets.
\item It works out of the box, with no user-specified tuning parameters.
\item It has strong statistical guarantees.
\end{enumerate}
These three factors make our proposed method a strong candidate to replace ordinary kernel density estimation as the default ``first pass'' for data-analysis practitioners.

The histogram trend-filtering estimator is related to the following variational optimization problem based on penalizing the log likelihood $g(x) = \log f(x)$:
\begin{equation}
\label{eqn:basic_variational_problem}
\begin{aligned}
& \underset{g}{\text{minimize}}
& &
- \sum_{i=1}^n  g(y_i) \\
& \text{subject to}
& &
\int_{\mathcal{R}} e^{g} = 1  \\
&
& & J(g) \leq t \, ,
\end{aligned}
\end{equation}
where $J(g)$ is a known penalty functional.  Imposing an appropriate penalty can encourage smoothness and avoids estimates that are sums of point masses.

Specifically, we consider solutions to (\ref{eqn:basic_variational_problem}) for penalties based on total variation, as proposed by \citet{koenker:mizera:2007}.   We provide conditions under which explicit rates of convergence can be obtained for these estimators.  We also study a finite-dimensional version of this variational problem---histogram trend filtering---which involves two conceptually simple steps.  First, partition the observations into $D_n$ histogram bins with centers $\xi_1 < \cdots < \xi_{D_n}$ and counts $x_1, \ldots, x_{D_n}$.  Then assume the surrogate model $x_j \sim \mathrm{Poisson}(\lambda_j)$ and estimate the $\lambda_j$'s via polynomial trend filtering \citep{kim:boyd:etal:2009,tibs:2014a} applied to the Poisson likelihood.  The renormalized $\lambda_j$'s then may be used to form an estimate of $f_0$.

Our results show that this simple, computationally efficient procedure yields excellent performance for density estimation.  Our main theorems characterize how the optimal bin size must shrink as a function of $n$  to ensure consistency for estimating $f_0$, and provide bounds on the proposed procedure's reconstruction error under the assumption that the bins are chosen accordingly.  Our empirical results also show that the histogram trend-filtering estimator is adaptive to changes in smoothness of the underlying density when familiar information criteria are used to choose the method's single tuning parameter.  Put simply, in can yield an estimate that is simultaneously smooth in some regions and spiky in others.  This behavior contrasts favorably with kernel density estimation, where the bandwidth parameter governs the global smoothness of the estimate.


\subsection{Histogram trend filtering}
\label{hist_trend_filtering}

The idea of histogram trend filtering is to reduce the density estimation problem to that of a nonparametric Poisson regression problem, which is solved by trend filtering \citep{kim:boyd:etal:2009,tibs:taylor:2011}.  The method is so computationally efficient for two reasons: (1) because binning the data results in a huge reduction from data points to bin counts, and (2) because the trend-filtering estimator for a Poisson regression can be obtained so cheaply, using the extraordinarily fast ADMM algorithm of \citet{ramdas:tibs:2014}.  An important point for us to demonstrate is that the data reduction step can be done without losing too much information; we address this concern later.

 Let us now construct in detail the histogram trend-filtering estimator, which can be viewed as a discrete approximation to Problem (\ref{eqn:basic_variational_problem}) when $J$ penalizes the total variation of $g$ or higher-order versions thereof.  We begin with several assumptions made for ease of exposition.  Let $\mathcal{X} \subset \mathcal{R}$ denote the  support of $f_0$. Suppose that $\mathcal{X}$ is a compact set that it is partitioned into $D_n$ disjoint intervals $I_j$ with midpoints $\xi_j$, such that $\bigcup_j I_j = \mathcal{X}$.  We assume that the intervals are of equal length $\delta_n$ and ordered so that $\xi_1 < \cdots < \xi_{D_n}$.  Any of these assumptions can be relaxed in practice. 

Now consider a histogram of the observations using bins $I_j$.  Let $x_j = \#\{ x_i \in I_j\}$ denote the histogram count for bin $j$, and consider the surrogate model
\begin{equation}
\label{eqn:surrogate_poisson_model}
x_j \sim \mbox{Poisson}(\lambda_j) \; , \quad \lambda_j =  n \delta_n f(\xi_j)  \approx  n \int_{I_j} f_0(y) \ d y \, .
\end{equation}
Let $\theta_j = \log \lambda_j$ be the log rate parameter for bin $j$, let $\theta = (\theta_1, \ldots, \theta_{D_n})$, and define the loss function
$$
l(\theta) = \sum_{j=1}^{D_n} \left\{
e^{\theta_j} - x_j \theta_j \right\}
$$
as the negative log likelihood corresponding to Model (\ref{eqn:surrogate_poisson_model}).  We propose to estimate $\theta$ using the solution to the unconstrained optimization problem
\begin{equation}
\label{eqn:histTF_objective}
\begin{aligned}
& \underset{\theta \in \mathcal{R}^D}{\text{minimize}}
& &
l(\theta) + \tau \Vert \Delta^{(k+1)} \theta \Vert^p_q \, ,
\end{aligned}
\end{equation}
where $\Delta^{(k+1)}$ is the discrete difference operator of order $k$.  Concretely, when $k=0$, $\Delta^{(1)}$ is the matrix encoding the first differences of adjacent values:
\begin{equation}
\label{eqn:fusedlassoD}
\Delta^{(1)} =\left(\begin{array}{rrrrrr}
1 & -1 & 0 & 0 & \mathbf{\cdots} & 0\\
0 & 1 & -1 & 0 & \cdots & 0\\
\vdots &  &  &  & \ddots & \vdots \\
0 & \cdots &  & 0 & 1 & -1
\end{array}\right).
\end{equation}
For $k \geq 1$ this matrix is defined recursively as $\Delta^{(k+1)} = \Delta^{(1)} \Delta^{(k)}$, where $\Delta^{(1)}$ from (\ref{eqn:fusedlassoD}) is of the appropriate dimension.

We focus on problem (\ref{eqn:histTF_objective}) when $q=p=1$, which corresponds to the polynomial trend-filtering estimator under a Poisson likelihood. Intuitively, the trend-filtering estimator is similar to an adaptive spline model: it places a lasso penalty on a discrete analogue of the order-$k$ derivative of the underlying log-density, resulting in a piecewise polynomial estimate whose degree depends on $k$.  Trend filtering has been studied extensively in the context of function estimation, generalized linear models, and graph denoising \citep{kim:boyd:etal:2009,tibs:taylor:2011,wang2014trend}.

The goal of this paper is to understand the statistical properties of this method as an approach to density estimation, and therefore we do not discuss details of implementation.  However, we note that problem (\ref{eqn:histTF_objective}) can be solved efficiently when $q=p=1$ using the augmented-Lagrangian method from \cite{ramdas:tibs:2014}, as implemented in the \verb|glmgen| R package \citep{glmgen:2014}.  When $q=p=2$, the objective is differentiable, and any standard gradient-based or quasi-Newton optimization method may be used.

\section{Connections with previous work}

In this section we present a brief review of density estimator related to our methods. We begin by discussing the seminal work from \cite{good:gaskins:1971} which can be motivated from a Bayesian perspective. This starts  by considering the  prior

$$
p(f)  \propto \exp \left ( - \Phi(f) \right ) \mathbb{I} \left ( f \in \mathcal{A} \right )
$$
where $ \Phi $ is a roughness penalty and $\mathcal{A} $ is some class of density functions.  Then, given the usual likelihood

\[
     p(x \mid  f)  = \prod_{i=1}^{n} f(x_i),
\]
the authors in \cite{good:gaskins:1971} produce  a maximum a posteriori (MAP) estimate of $f_0$  by solving
\begin{equation}
\label{MAP_problem}
\hat{f} = \text{argmin}_{ f \in \mathcal{A} } - \log p( x \mid f ) + \Phi (f) \, .
\end{equation}
This is the main focus of study in \cite{good:gaskins:1971},  where one of the choices of roughness penalty is proportional to Fisher's information concerning the displacement or location, regarded as a parameter:

\begin{equation}
\label{penalty}
   \Phi(f)   = \int_{-\infty}^{\infty} \frac{(f^{\prime}(x))^{2}}{f(x)}dx
\end{equation}

The  consistency properties of the estimator (\ref{MAP_problem}) were briefly studied in \cite{good:gaskins:1971}, where the authors showed convergence in the sense of probability of integrals of the form

\[
   \int_{a}^{b} \hat{f}(x)dx   \rightarrow  \int_{a}^{b} f_0(x)dx.
\]
However, this does not imply the absence of false bumps: they could become small and numerous as $n$ increases.   A  more complete characterization of the estimator $\hat{f}$  with the choice of penalty (\ref{penalty}) was given in \cite{de1975nonparametric}. There, the main result is the proof of the existence and uniquenes of $\hat{f}$. Moreover, the result holds with more generality allowing the class functions $\mathcal{A}$ to be a reproducing Hilbert space, and stating that if $\Psi$ is the square of the norm of such reproducing space, then $\hat{f}$  exists and it is unique.

In  a variation of the estimator from \cite{good:gaskins:1971},  \cite{silverman1982estimation} works within the framework  of roughness penalty. However, rather than penalties directly imposing constraints on the density space,  \cite{silverman1982estimation} proposes to penalize the log density. This is immediately attractive since it automatically imposes a positive constraint in the estimates, with the formal formulation given as

\begin{equation}
\label{silverman}
\begin{array}{ll}
    \underset{ g }{ \text{minimize} } &  -\frac{1}{n}\,\sum_{i= 1}^{n} g(x_i) + \frac{1}{2} \lambda\,\Phi(g)\\
      \text{subject to} & \int e^{g(\mu)}d\mu   = 1.
 \end{array}
\end{equation}
 The roughness penalties studied in \cite{silverman1982estimation} are of the form

 \[
     \Phi(g) =  \int_{0}^{1} \left[D(g )(\mu)\right]^2\,d(\mu)
 \]
where $D(g)$ is a function of the first $m$ derivatives of $g$, see \cite{silverman1982estimation} for the specific construction. There, Theorem 3.1 also shows that (\ref{silverman}) is equivalent to the unconstrained problem

\begin{equation}
\label{silverman2}
    \underset{ g }{ \text{minimize} } \,  -\frac{1}{n}\,\sum_{i= 1}^{n} g(x_i) + \frac{1}{2} \lambda\,\Phi(g) +  \int e^{g(\mu)}d\mu.  \\
\end{equation}
This alternative formulation has the nice feature that can be formulated as a convex optimization problem, see \cite{o1988fast}.

It turns out that a similar result can easily be proven for our Poisson surrogate problem. This is given in the following Theorem.

\begin{theorem}
\label{equivalent_formultion}
With the notation from Section \ref{hist_trend_filtering}, it can be proven that there exists a constant $c>0$  such that $\hat{\theta}$ solves (\ref{eqn:histTF_objective})  if only if $\hat{g}_i =  \hat{\theta}_i - \log(n\,\delta_n) $   solves
\begin{equation}
\label{constrained_problem}
\begin{array}{ll}
    \underset{ g }{ \text{minimize} } &  -\frac{1}{n}\,\sum_{i= 1}^{D_n} x_i\,g_i   \\
      \text{subject to} & \sum_{i=1}^{D_n} \delta_n\,e^{g_i}   = 1,\,\,\,\,   \| \Delta^{(k+1)}g \|_q^p \leq c. \\
 \end{array}
\end{equation}
\end{theorem}
Thus,  we have shown that our Poisson surrogate problem is indeed a discretization of problem (\ref{eqn:basic_variational_problem}), where we replace the classical likelihood by a cross entropy objective, the integrability constraint by a constraint on the rectangle rule for the estimator, and the total variation penalty by a discrete version using difference matrices.

While our histogram trend filtering approach to density estimation might seem closely related to the estimator from  \cite{silverman1982estimation}, there are two significant differences. First,  as pointed out by  \cite{sardy2010density},  the estimator given by problem (\ref{silverman}) tends to over-smooth, since non-smoothness is penalized more heavily at low density values than at high density values, which may lead to uneven smoothing. The authors in \cite{sardy2010density} address this problem by imposing a a total variation penalty. Thus, giving rise to the problem

\begin{equation}
\label{sardy}
\begin{array}{ll}
    \underset{ f }{ \text{minimize} } &  -\frac{1}{n}\,\sum_{i= 1}^{n}\log(f_i) +  \lambda\,\sum_{i=2}^{n}\vert f_i - f_{i-1} \vert   \\
      \text{subject to} &   a^T\,f   = 1\\
 \end{array}
\end{equation}
for some integration coefficient vector $a$ and parameter $\lambda > 0$. The main motivation for this problem is to avoid the over-smooth solutions from solving (\ref{silverman}). Hence, given the flexibility of imposing $\|\cdot\|_q^p$, our histogram trend filtering estimators are also expected to avoid over-smoothing  by taking $p=q =1$.  However, the other important issue associated with the estimator given by (\ref{silverman}) is the computational complexity. This is also shared by the estimator from \cite{sardy2010density} since both of these procedures require to estimate a vector in $R^n$. In contrast, we solve optimization problems in a significantly lower dimensional space, $R^{D_n}$.

Next we observe that by its mere definition in Problem \ref{eqn:histTF_objective}, when $p=q=1$, our density estimator  provides piecewise polynomial solutions in the log-space. Here, the parameter $k$ in the difference matrix indicates the degree of the polynomial approximation used. For instance, $k=0$  corresponds to piecewise constant solutions,  while $k=1$ to piecewise linear solutions. An attractive feature of our method is that it is not necessary to specify the the locations of break points; this is done adaptively by solving a convex optimization problem. In contrast,  \cite{barron1991approximation} consider fitting splines in the log-space but this requires specification of the locations of the of the knots. Moreover, \cite{barron1991approximation} provides rates of convergence for such spline estimators, in terms of the Kullback-Leiber divergence, when the true density satisfies
\begin{equation}
\label{splines}
   \int \vert\left(\log f_0\right)^{(k+1)}\vert^{2} <  \infty
\end{equation}
with the superscript $(k+1)$  denoting the ($k+1$)-the derivative. While we do not explicitly require this condition on the true density for the subsequent analysis,  we do work we spaces of densities  for which

\[
      \int  \vert\left(\log f\right)^{(k+1)}\vert^{q}  < \infty,
\]
and $q \in \{1,2\}$. Thus our framework includes the less restrictive case $p= 1$. In fact, when $p=1$, our result in Theorem \ref{theorem_2}  provides convergence rates for our discrete estimator and this directly incorporates the smoothness of the true log-density captured by a difference operator. This is of interest given that structure is lost when we move from $p=2$ to $p=1$ since $L_1$ does not come with an inner product.

Finally, we review  the penalized estimator from \cite{willett2007multiscale}.  This is obtained by solving the problem

\begin{equation}
\label{willet}
\begin{array}{llll}
 \hat{f}_{W}  &  = &  \underset{f}{\text{arg min}} &  -\frac{1}{n}\,\sum_{i= 1}^{n}\log(f(y_i)) +  \text{pen}(f)\\
  &  & \text{subject to }  &  \int f = 1,\,\, f\in \mathcal{C}
\end{array}
\end{equation}
where $\mathcal{C}$ is a class of non-negative piecewise polynomials and $\text{pen}(f)$ is a functional that penalizes the complexity of polynomials. The solution to (\ref{willet}) enjoys attractive theoretical properties,  \cite{willett2007multiscale}  shows that if $f_0 $  is a member of the Besov space $B_{q}^{\alpha}\left( L_{p}([0,1])\right)$
 where $\alpha > 0$, $1/p = \alpha + 1/2$ and $0 < p < q$, then,

 \[
      E\left[  \| f_0^{1/2} -  \hat{f}_{W}^{1/2}  \|_2^2 \right] \leq  C\left(  \frac{\log_2^2(n) }{n} \right)^{\frac{2\alpha}{2\alpha +1}}.
 \]
 Moreover, the estimator $\hat{f}_W$ involves using  recursive dyadic partitions in order
to produce near-optimal, piecewise polynomial estimates,
analogous to the methodologies in \cite{breiman1984classification,kolaczyk2004multiscale}  and \cite{donoho1997cart}. Also, this  multiscale method provides spatial adaptivity similar to
wavelet-based techniques \citep{donoho1995wavelet,kerkyacharian1996lp}, with a notable advantage.
Wavelet-based estimators can only adapt to a function’s
smoothness up to the wavelet’s number of vanishing moments;
thus, some a priori notion of the smoothness of the true
density or intensity is required in order to choose a suitable
wavelet basis and guarantee optimal rates. The estimator $\hat{f}_W$, in contrast, automatically adapts to arbitrary degrees
of the function’s smoothness without any user input or prior information.  However, this penalized method requires elevated computational effort. Specifically, it it involves  $O(n\,\log_2(n))$ calls to a convex
minimization routine and $O\left(n\,\log_2 (n)\right)$ comparisons of the resulting
(penalized) likelihood values. The goal of this paper is to provide a computationally  efficient estimator that can adapt to different smoothness of the true density and that comes with statistical guarantees.

\section{Main results}

\subsection{Variational formulation and  rates of convergence}

 Now we present rates of convergence based on viewing (\ref{eqn:histTF_objective}) as a variational problem. To that end, we observe that (\ref{eqn:histTF_objective}) can be thought as a discrete approximation to

\[
\begin{array}{ll}
    \underset{g}{\text{minimize}} & - \sum_{i=1}^{D_n} x_i\,g(\xi_i)\\
     \text{   subject to }  &  \int \vert g^{(k+1)}(\mu) \vert^q  d\mu \leq t   \\
     &    \,\,\int e^{g(\mu)} \ d \mu = 1.
\end{array}
\]

This formulation is more general than that of (\ref{eqn:basic_variational_problem}), since one can always choose the number of points in each bin to be equal to one and replace the mid-points $(\xi_j)_{j=1}^{D_n}$ by the actual observations   $(x_i)_{i=1}^{n}$.   When $n$ is relatively small, we recommend just this.  However, when $n$ is large this becomes computationally burdensome.  Since we should expect to lose information by considering counts within bins instead of the observations, it is natural to ask if we can provide error bounds for both situations and to see how these  differ. Our next theorem gives light on this point.

To state such result, we first introduce some notation and make some assumptions. These assumptions are designed to ensure that the set over which we constrain the minimization is compact with respect to the supremum norm. The idea of proving consistency results
for non-parametric maximum likelihood estimators over a sequence of reduced spaces is known as the method of sieves.  This technique was introduced by \cite{geman1982nonparametric} and has been further studied in the literature in \cite{birge1998minimum}, \citet{shen1994convergence}, and \citet{shen1997methods}.

Given a function $h$ with domain $\Omega$, we say that $h$ is  $L$-Lipschitz if $|h(x) - h(y)| \leq L\,\vert x - y \vert^{\alpha} $  for all $x,y \in \Omega$.  The order-$l$ weak  derivative of $h$ is denoted by $h^{(l)}$. The set of log-Sobolev densities in $\Omega  = (0,1)$ is

$$\mathcal{P} := \left\{ h \,\,:\,  \,\,h \in W^{k+1,p}\left(\Omega\right),\,\, \,\, \int_{0}^{1} e^{h(\mu)}d\mu = 1, \right\}$$

where $d\mu$ denotes the Lebesgue integral and $W^{k+1,p}\left(\Omega\right)$ is the Sobolev space of order $k+1,p$ in $\Omega$. Recall from \cite{oden2012introduction}  that  $W^{m,p}(\Lambda)$, the Sobolev space of order $m,p$ on a bounded set $\Lambda \subset \mathbb{R}$, is defined as the set of functions  $u :  \Lambda \rightarrow \mathbb{R}$ such that

\[
   \|u^{(\alpha)}\|_{L^p(\Lambda)} := \left(\int_{\Lambda} \vert u^{(\alpha)} \vert^p\right)^{1/p}  < \infty
\]
for all $\alpha \in \{0,1,\ldots,m\}$, where $u^{(\alpha)}$   is the $\alpha$-th weak derivative of $u$. 

On the other hand, for an open set $\Lambda \subset \mathbb{R}$, we denote by $\bar{C}(\Lambda)$ as the set of continuous function with the finite suppremum norm

\[
  \|u\|_{L^{\infty}(\Lambda)} = \underset{x \in \Lambda}{\sup } \vert u(x) \vert.
\]
Note that the $\bar{C}(\Lambda)$ do not require its elements to be uniformly continuous but only the relaxed condition that they are bounded in the open set $\Lambda$. Finally, for a set $S \subset \bar{C}(\Lambda)$ we denote by $\text{cl}_{\Lambda}(S)$ its closure in $\bar{C}(\Lambda)$ with respect to $ \|\cdot\|_{L^{\infty}(\Lambda)}$.


\begin{assumption}\label{as:1}
 We consider $l_n$, $T_n$   be positive sequences of numbers to be defined later. For now we only assume that these sequences are bounded by below.
\end{assumption}

\begin{assumption}\label{as:2}
We assume that the true density $f_0$ has support in $[0,1]$ and  is $r$-Lipschitz for some positive constant $r$.
\end{assumption}


\begin{definition}\label{as:3}

We define the sets $S_{n,1}$ and $S_{n,2}$ as
\[
\begin{array}{lll}
S_{n,i} & = & \left\{e^h \,:\,   h  \in  logS_{n,i} \right\}
\end{array}
\]
with
\[
  logS_{n,i} = \text{cl}_{\Omega}\left(  \mathcal{P} \cap \left\{ h\, :\,\,  \left\|h\right\|_{L^\infty(\Omega)}\leq T_n,\,\,\left\|h^{(k+1)}\right\|_{L^{i}(\Omega)}^i  \leq T_n,  \,\,\,h \,\,\,\text{is continuous}  \right\}  \right).
\]
\end{definition}

\begin{definition}\label{as:4}
For $i=1,2$  we construct

\[
\begin{array}{lll}
S_{n,i}^{\prime} & = & S_{n,i} \cap \left\{e^h \,:\,\,\,\,  h \in \bar{C}(\Omega), \,\, h \,\,\text{ is  } r\,e^{-T_n}\text{-Lipschitz},\,\,e^h \geq l_n \,\right\}
\end{array}
\]
\end{definition}

Using the assumptions and notation above, we study $M_{n,i}$,  the solution set of the problem

\begin{equation}
\label{constrained_problem_sn}
\begin{aligned}
& \underset{g}{\text{minimize}}
& &
- \sum_{j=1}^{n} \left\{  \log \left(f\right)\left(y_j\right)  \ d \mu   \right\} \\
& \text{subject to}
& &  f \in S_{n,i},
\end{aligned}
\end{equation}
for $i = \{1,2\}$. We also consider the case when $S_{n,i}$ is replaced by $S_{n,i}^{\prime}$  and the objective is replaced by a weighted likelihood

\begin{equation}
\label{constrained_problem_sn_prime}
\begin{aligned}
& \underset{g}{\text{minimize}}
& &
- \sum_{j=1}^{D_n} \left\{  x_j\,\log \left(f\right)\left(\xi_j\right)  \ d \mu   \right\} \\
& \text{subject to}
& &  f \in S_{n,i}^{\prime},
\end{aligned}
\end{equation}

in which case the respective solution set is denoted by  $M_{n,i}^{\prime}$. We proceed to state convergence rates proved using entropy techniques as in \cite{van1990estimating,mammen1991nonparametric,wong1995probability,ghosal2001entropies}.


\begin{theorem}
\label{convergence_rates}

\begin{description}

\item[Part A.]  Assume that $T_n\,e^{T_n} = O((\log n)^{4q} )$ where $q> 1/2$. If there exists a sequence $q_{n,2} \in S_{n,2}$  and

\[
    \int_{\Omega} f_0(\mu)\{q_{n,2}(\mu)\}^{-1/2}\left[f_{0}(\mu)^{1/2} - q_{n,2}(\mu)^{1/2}  \right]d\mu = O(1/n).
\]
Then, $M_{n,2}^{\prime}$ is not empty, and if $D_n =  n^{1/s}$ for $1< s $  with $l_n\,n^{1/s}\, \geq n^{\alpha}$, $0 < \alpha < 1/s$, then there exists positive constants $c_1$, $c_2$ and $C_2$  (independent of $f_0$) such that for large enough $n$,
\[
  \begin{array}{lll}
   \text{P}^{*}\left[ \underset{\hat{f}_n \in M_{n,2}^{\prime}}{\text{sup }} d\left(\hat{f}_n,f_0\right) \geq C_2\{\text{log}(n)\}^q\,n^{-\alpha/2}     \right] & \leq  &  c_2\,\text{exp}\left[\,-c_1\,\{\text{log}(n)\}^{2q}\,n^{1-\alpha} \right].
  \end{array}
\]
Moreover, if we replace $S_{n,2}$  by $S_{n,2}^{\prime}$, and set $\alpha = 1$, we obtain the same concentration bound for $M_{n,2}$.
\item[Part B.]  Suppose that $D_n = n^{1/s}$  with $s>1$, $\left( \,T_n\,e^{T_n}\right)^{1/(2k+2)} =  O(n^{r})$ for some $r \in (0,1/2)$, and there exists $q_{n,1} \in S_{n,1}$  such that

    \[
    \int_{\Omega} f_0(\mu)\{q_{n,1}(\mu)\}^{-1/2}\left[f_{0}(\mu)^{1/2} - q_{n,1}(\mu)^{1/2}  \right]d\mu = O(1/n).
\]
If $D_n =  n^{1/s}$ for $1< s $  with $l_n\,n^{1/s}\, \geq  n^{\alpha}$, $0 < \alpha < 1/s$, then,  for all
\[
0 < t <  \min\left\{\alpha/2 , \frac{1/2 - r}{ 1 + 2/(k+1)} \right\},
\]
we have that $M_{n,1}^{\prime} \neq \emptyset$  and

    \[
       \text{P}^{*}\left( \underset{\hat{f}_n \in M_{n,1}^{\prime}}{\text{sup }} \sqrt{H^{2}\left( f_0^{1/2} , f_n^{1/2} \right)}   > C_1 \frac{1}{n^{t  }} \right)  \leq  e^{-n^{1-2t}},
    \]
        with $C_1$  a positive constant independent of $f_0$. Moreover, the same conscentration bound holds for $M_{n,1}$ without requiring the conditions imposed by $\alpha$.


\end{description}
\end{theorem}

 The existence of the sequences $q_{n,i}$ in  Theorem \ref{convergence_rates} is meant to impose regularity conditions on the true density $f_0$ that allow it to be sufficiently well approximated by elements of the sieve. Thus,we do not require that $f_0$ a smooth function but rather that can be well approximated by our sieves. On the other hand, we can think of the sieves as class of continuous functions that are well approximated by smooth functions. In $S_{n,i}$  the constraint given by $\|\cdot\|_{L^i(\Omega)}^i$ is designed to enforce smoothness. Moreover, taking

 $$
 T_n = \,2\,q\,\log(\log(n)) + 2^{-1}\log(C)
 $$
 with a constant $C>0$ ensures that the elements of $S_{n,2}$ will take values in

 $$((\log n)^{-2q}\,C^{-1/2},(\log n)^{2q}\,C^{1/2})$$
 which approaches to $(0,\infty)$ as $n \rightarrow \infty$. A similar statement can also be made about the sieve $S_{n,1}$.

 Note that when $T_n $ is  constant and  $f_0 \in S_{n,2}$, then, the first part of Theorem \ref{convergence_rates} implies

 \[
    \underset{n \rightarrow \infty}{ \lim }\underset{f_0 \in S_{n,2}}{ \sup}   \text{P}^{*}\left[ \underset{\hat{f}_n \in M_{n,2}^{\prime}}{\text{sup }} d\left(\hat{f}_n,f_0\right) \geq C_2 \{\text{log}(n)\}^q\,n^{-1/2}     \right] \rightarrow 0.
 \]
A related  result for log-spline density estimation  was found in \cite{barron1991approximation} for the Kullback-Leibler divergence when $f_0$  in the log-space space belongs to a Sobolev ball. This then implies that in terms of the Hellinger distance, the estimator $\hat{p}_n$ from  \cite{barron1991approximation}  satisfies
\[
\underset{K \rightarrow \infty}{ \lim }\, \underset{n \rightarrow \infty}{ \lim }\,\underset{f_0 \in B}{ \sup}\,   \text{P}\left[  d\left(\hat{p}_n,f_0\right) \geq K n^{-\frac{2k+2}{2(2k+3)}}   \right] \rightarrow 0.
\]
where
\[
   B = \left\{ f :  \max\left\{ \|\log(f^{(k+1)})\|_{L^2([0,1])}, \|\log(f)\|_{L^\infty([0,1])} \right\} \leq c \right\}
\]
for a positive constant $c$. Thus, Theorem \ref{convergence_rates} provides surprising convergence results for the sieve $S_{n,2}$. Moreover, the second part of Theorem \ref{convergence_rates} addresses the case in which the true density can be approximated by functions in a sobolev ball in $W^{k+1,1}$. This differs from previous work given that $L^1$ spaces are not Hilbert spaces, and hence the analysis from \cite{silverman1982estimation} is not applicable.



\subsection{Discrete estimator}

In the previous subsection we characterized the variational estimator corresponding to histogram trend filtering. Next we focus on the version of the estimator where we approximate the function on the discrete grid. This is necessary for finding the solution by numerical optimization, and is analogous to the approach taken by \cite{koenker:mizera:2007}, who started with a variational problem and then moved to an approximate solution on a grid. However, they did not provide any statistical guarantees for either of their formulations.

We  now denote  the regularization parameter as $\tau_n$ and define the vectors

\begin{equation}
\label{true_parameter}
   \begin{array}{l}
 \theta^0 := \left\{\text{log}n - \text{log}D_n + \text{log}f_0(\xi_1),\ldots,\text{log}n - \text{log}D_n + \text{log}f_0(\xi_{D_n})\right\}\\
   \end{array}
\end{equation}
and $\hat{\theta}$ as the solution to Problem  (\ref{eqn:histTF_objective}).  Thus up to a known constant of proportionality, $\theta^0$ and $\hat \theta$ are the true and estimated log densities, respectively.

We are now ready to state our next consistency result. Its proof can be found in the appendix and is inspired by previous work on concentration bounds for estimators formulated as convex-optimization problems \citep[e.g.][]{ravikumar2010high,tansey2015vector}.

\begin{theorem}
\label{theorem_1}
Let  $p=q=1$  or  $p = q = 2$ and $s>2$ and assume that the true density $f_0$ is  $L$-Lipschitz for some constant $L > 0$ and  $ \text{Suppored}(f_0)  = [0,1]$. Suppose that we choose $D_n = \Theta\left(n^{1/s}\right)$.  Then there exists a constant $ r > 0$ and a function $\phi$  satisfying

\[
    0  <  \underset{n \rightarrow \infty}{\text{lim }} \,\, n^{-1 + 2/s} \phi(n)  <  \infty,
 \]
 such that if we choose  $\tau_n  \leq  r\,n^{1- 3/(2s)}$, then
\[
    \text{P} \left(  \frac{\|\theta^0 -\hat{\theta}\|_2^2}{D_n} \geq  \frac{1}{n^{1/s}}   \right) \leq \text{exp}\left\{-\phi(n)\right\}.
\]
\end{theorem}

Theorem \ref{theorem_1} is quite general, in that it establishes consistency under the assumption that the penalty parameter is relaxed at a sufficient rate.  However, it does not explicitly incorporate the role of the penalty function in (\ref{eqn:histTF_objective}) in establishing the accuracy of the method.  Our final result provides a bound on estimation error that does refer to the penalty explicity. This, unlike Theorem \ref{theorem_1}, can be extended to densities of unbounded support.

\begin{theorem}
\label{theorem_2}
Let $\xi_{j}^{\prime}  $  be the point  in $I_j$  satisfying

\[
   \delta_n\,f_0(\xi_j^{\prime}) = \int_{I_j} f_0(t)dt.
\]
Let us also take $b \in (0,1/2)$ and define $\hat{\theta}$  as the solution to the convex optimization problem
\begin{equation}
\label{constrained_problem2}
\begin{array}{ll}
    \underset{ \theta }{ \text{minimize} } &  \sum_{j=1}^{D_n} \left\{
e^{\theta_j} - x_j \theta_j \right\} + \frac{\tau}{2}\| \Delta^{(k+1)}\theta \|_1   \\
      \text{subject to} &  \vert \theta_j  -  log(n\,\delta_n)\vert \leq n^{b},\,\,j=1,\ldots,D_n.\\
 \end{array}
\end{equation}
Let us assume that $(f_0(\xi_j^{\prime}),\ldots,f_0(\xi_{D_n}^{\prime}))$ belongs to the constraint set of \ref{constrained_problem2}.  Then
\[
 \hat{f}(\xi_j^{\prime}) =  \frac{  \exp\left( \hat{\theta}_j \right) }{n\,\delta_n},\,\,\, j=1,\ldots,D_n,
\]
satisfies

\begin{equation}
\label{bound}
     \sum_{j=1}^{D_n} \delta_n\,  f_0(\xi_j^{\prime}) \log\left(\frac{f_0(\xi_j^{\prime})}{\hat{f}(\xi_j^{\prime})}  \right) = O_{\text{P}}\left( \frac{ \|\left(\Delta^{(k+1)}\right)^{-}\|_{\infty}  }{D_n^{r}}\|\Delta^{(k+1)}\log\left(f_0(\xi^{\prime})\right) \|_1 + \frac{1}{n^{1/2-b}} \right),
\end{equation}
where we choose $D_n = \Theta\left(n^{1/s}\right)$ and  $\tau = \Theta\left( n^{1-r/s}\|\left(\Delta^{(k+1)}\right)^{-}\|_{\infty}\right)$,   where  $s > 1$ and $r \in (0,s/2)$.
\end{theorem}

Theorem \ref{theorem_2}  states convergence rates for our Poisson surrogate model estimator. In particular, the bound  in (\ref{bound}) controls the Kullback–-Leibler divergence   between our estimator and a discretized version of the true density. Moreover,  the constraint on the supremum norm in the log-space space ensures that the optimization is over a compact set. Hence,  it is not restrictive given that this bound tends to infinity  as $n$ increases.

Finally, we emphasize that $ \|\left(\Delta^{(k+1)}\right)^{-}\|_{\infty}  = O(D_n)$. This is a consequence of the proof of Corollary 4 in \cite{wang2014trend}. In practice we have found that  $\|\left(\Delta^{(k+1)}\right)^{-}\|_{\infty}\,D_n^{-1} \in (.1474,.1482)$ if if $D_n$  is chosen between $500$  and $10000$.

\subsection{Model selection}
\label{model_section}

We know turn to the discussion of the parameters $D_n$ and $\tau$  when $p=q=1$. For the former of these parameters, we see from the results in the previous section that $D_n = O(n^{1/s})$,  for $s > 2$, seems a reasonable choice. In practice we have observed that the rule $D_n = 10\,n^{1/2.5}$ performs excellently and hence it becomes our default choice. On the other hand, for the choice of $\tau$, we see from Theorem \ref{theorem_2} that
\[
\tau = \Theta\left( n^{1-r/s}\|\left(\Delta^{(k+1)}\right)^{-}\|_{\infty}\right)
\]
is a candidate choice with $r$ satisfying the constraint from Theorem \ref{theorem_2}. We will see in the next section that this performs well in practice as well.  In particular, this choice ensures consistency for the trivial case in which the true density $f_0$ is uniform, the precise definition of the universal penalty required in  \cite{sardy2010density}.

Finally, on the choice of $\tau$, we also consider an add-hoc  rule inspired by the work of \cite{tibshirani2012degrees} on regression problems with generalized lasso penalties. This consists of computing the solution path of the problem \ref{eqn:histTF_objective} and  then considering a surrogate AIC approach by computing
\[
    \text{AIC}_{\tau} = l(\hat{\theta}_{\tau}) + k  + 1 + \left\vert  \left\{ i :  (\Delta^{(k+1)}\hat{\theta}_{\tau} )_i  \neq 0\right\} \right\vert.
\]
The parameter $\tau$ is then chosen to minimize the expression above.

\section{Examples and discussion}

\subsection{Comparison with kernel methods}

We conducted a  simulation study to examine the performance of histogram trend filtering versus some common methods for density estimation.  Our first example is a three-component mixture of normals
$$
f_1(y) = 0\point9 N(y \mid 0, 1) + 0\point1 N(y \mid -2, 0\point1^2) + 0\point1 N(y \mid 3, 0\point5^2)
$$
shown in the top left panel of Figure \ref{fig:sim_examples}.  The second example is a five-component mixture of translated exponentials:
$$
f_2(y) = \sum_{c=1}^7 w_c \ Ex(y - m_c \mid 2) \, ,
$$
where the weight vector is $w = (1/7, 2/7, 1/7, 2/7, 1/7)$ and the translation vector is $m = (-1, 0, 1, 2, 3)$.  Here $Ex(y \mid r)$ means the density of the exponential distribution with rate parameter $r$.  This density is shown in the top right panel of Figure \ref{fig:sim_examples}.

\begin{figure}[bp!]
\begin{center}
\includegraphics[width=5.5in]{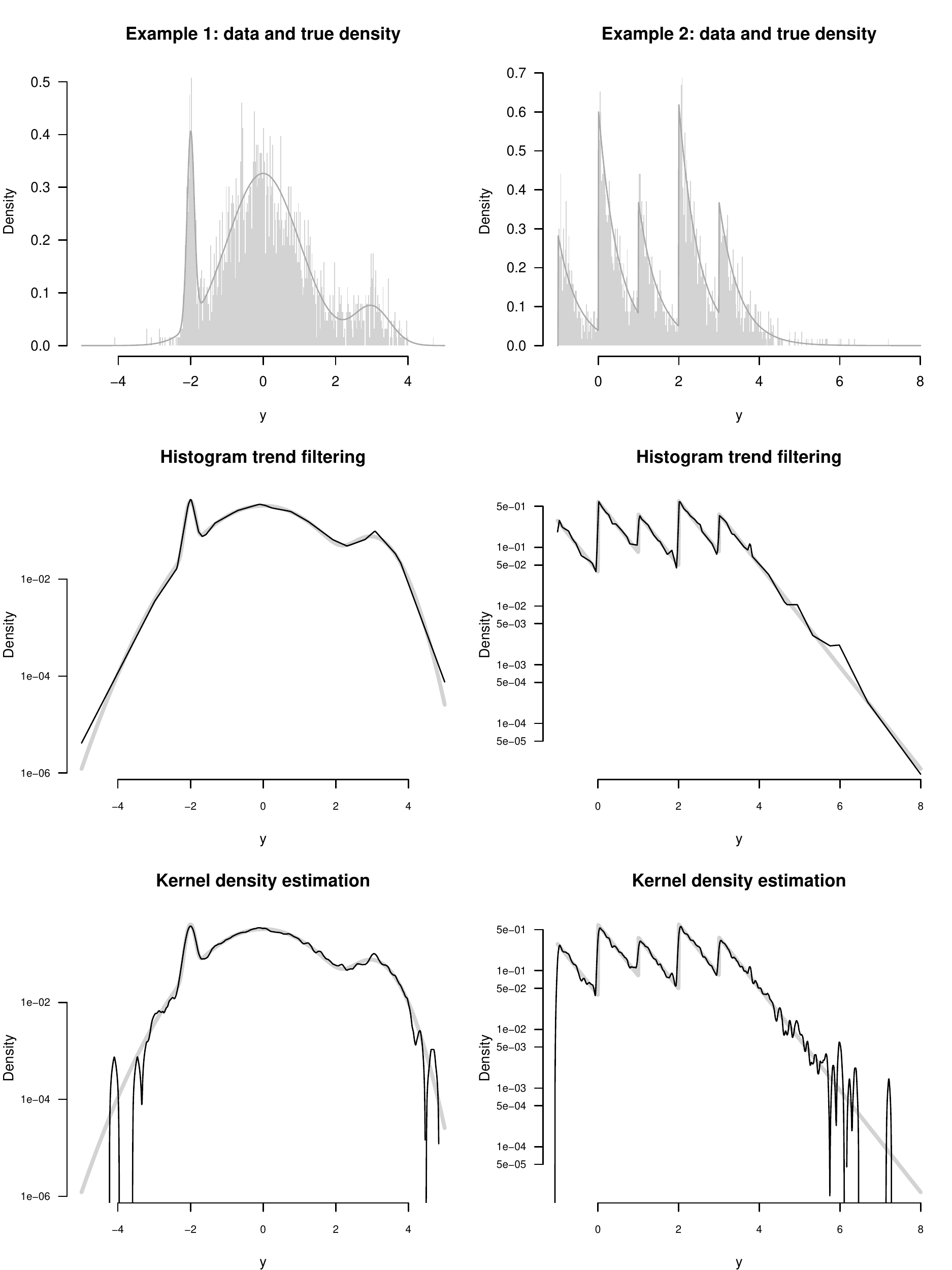}
\caption{\label{fig:sim_examples} Top two panels: the true densities $f_1$ (left) and $f_2$ (right) in the simulation study, together with samples of $n=2500$ from each density.  Middle two panels: results of histogram trend filtering for the $f_1$ sample (left) and the $f_2$ sample (right).  Bottom two panels: results of kernel density estimation for the $f_1$ sample (left) and the $f_2$ sample (right).  In the bottom four panels the reconstruction results are shown on a log scale.}
\end{center}
\end{figure}

Our simulation study consisted of 25 Monte Carlo replicates for each of six different sample sizes: $n=$ 500, 1000, 2500, 5000, 10000, and 50000.  For each simulated data set, we ran histogram trend filtering with $k=1$ and $k=2$.  We benchmarked the approach against three other methods: kernel density estimation with the bandwidth chosen by five-fold cross-validation, kernel density estimation with the bandwidth chosen by the normal reference rule, and local polynomial density estimation with smoothing parameter chosen by cross-validation.  In the reference-rule version of kernel density estimation, the bandwidth is chosen to be 0.9 times the minimum of the sample standard deviation and the interquartile range divided by $1.06 n^{-1/5}$ \citep{scott:1992}.  We used the version of local polynomial density estimation implemented in the R package \verb|locfit|.

\begin{table}[t]
\centering
\caption{\label{tab:sim1} Mean-squared error $\times$ 100 on example 1 for histogram trend filtering with $k=1$ and $k=2$ versus three other methods: kernel density estimation with bandwidth chosen by cross-validation, kernel density estimation using the normal reference rule, and local polynomial density estimation.}
\medskip
\begin{small}
\begin{tabular}{p{3pc} p{5pc} p{5pc} p{5pc} p{5pc} p{5pc}}
n & HTF ($k=1$) & HTF ($k=2$) & KDE (CV) & KDE (ref) & LP \\
500 & 2.5 & 4.9 & 3.1 & 4.0 & 3.3 \\
  1000 & 1.8 & 2.8 & 2.2 & 3.8 & 2.3 \\
  2500 & 1.3 & 1.6 & 1.7 & 3.3 & 1.6 \\
  5000 & 1.1 & 1.1 & 1.3 & 3.1 & 1.2 \\
  10000 & 0.7 & 0.7 & 0.9 & 2.8 & 0.9 \\
  50000 & 0.3 & 0.3 & 2.5 & 2.2 & 0.4 \\
\end{tabular}
\end{small}
\end{table}

\begin{table}[t]
\centering
\caption{\label{tab:sim2} Mean-squared error $\times$ 100 on example 2 for the same five methods in Table \ref{tab:sim1}.}
\medskip
\begin{small}
\begin{tabular}{p{3pc} p{5pc} p{5pc} p{5pc} p{5pc} p{5pc}}
n & HTF ($k=1$) & HTF ($k=2$) & KDE (CV) & KDE (ref) & LP \\
500 & 5.7 & 6.8 & 5.5 & 8.8 & 6.2 \\
  1000 & 4.0 & 4.6 & 4.5 & 8.5 & 4.9 \\
  2500 & 3.0 & 3.3 & 3.7 & 7.9 & 3.5 \\
  5000 & 2.4 & 2.9 & 3.2 & 7.6 & 2.9 \\
  10000 & 2.0 & 2.9 & 2.8 & 7.0 & 2.6 \\
  50000 & 1.6 & 2.9 & 6.1 & 5.9 & 2.3 \\
\end{tabular}
\end{small}
\end{table}

Tables \ref{tab:sim1} and \ref{tab:sim2} show the average mean-squared error of reconstruction of all methods for both $f_1$ and $f_2$.  Order-$1$ trend filtering has the lowest mean-squared error across all situations.  Figure \ref{fig:sim_examples} provides a detailed look at the two simulated data sets.  The top two panels show $f_1$ and $f_2$ together with a single simulated data set of $n=2500$ from each density.  The middle two panels show the reconstruction results for histogram trend filtering with $k=1$, while the bottom two panels show the reconstruction results for kernel density estimation with the bandwidth chosen by cross validation.  The trend-filtering estimator shows excellent adaptivity: it captures the sharp jumps in each of the true densities, without suffering from pronounced undersmoothing in other regions.

\subsection{Comparison with other penalized methods}

In the two previous examples we have considered comparisons versus estimation methods that scale well with the number of samples. We now conclude with an example comparing our histogram trend filtering versus other penalized methods that face problems with large numbers of samples. These methods are the the penalized likelihood approach from \cite{willett2007multiscale} (W-N), and the   total variation approach from \cite{sardy2010density} (TV) using their universal penalty. We also compare against the taut string method from \cite{davies2004densities}, which is closely related to the estimator from  \cite{sardy2010density}.

\begin{figure}[h!]
\begin{center}
\includegraphics[width=5.3in]{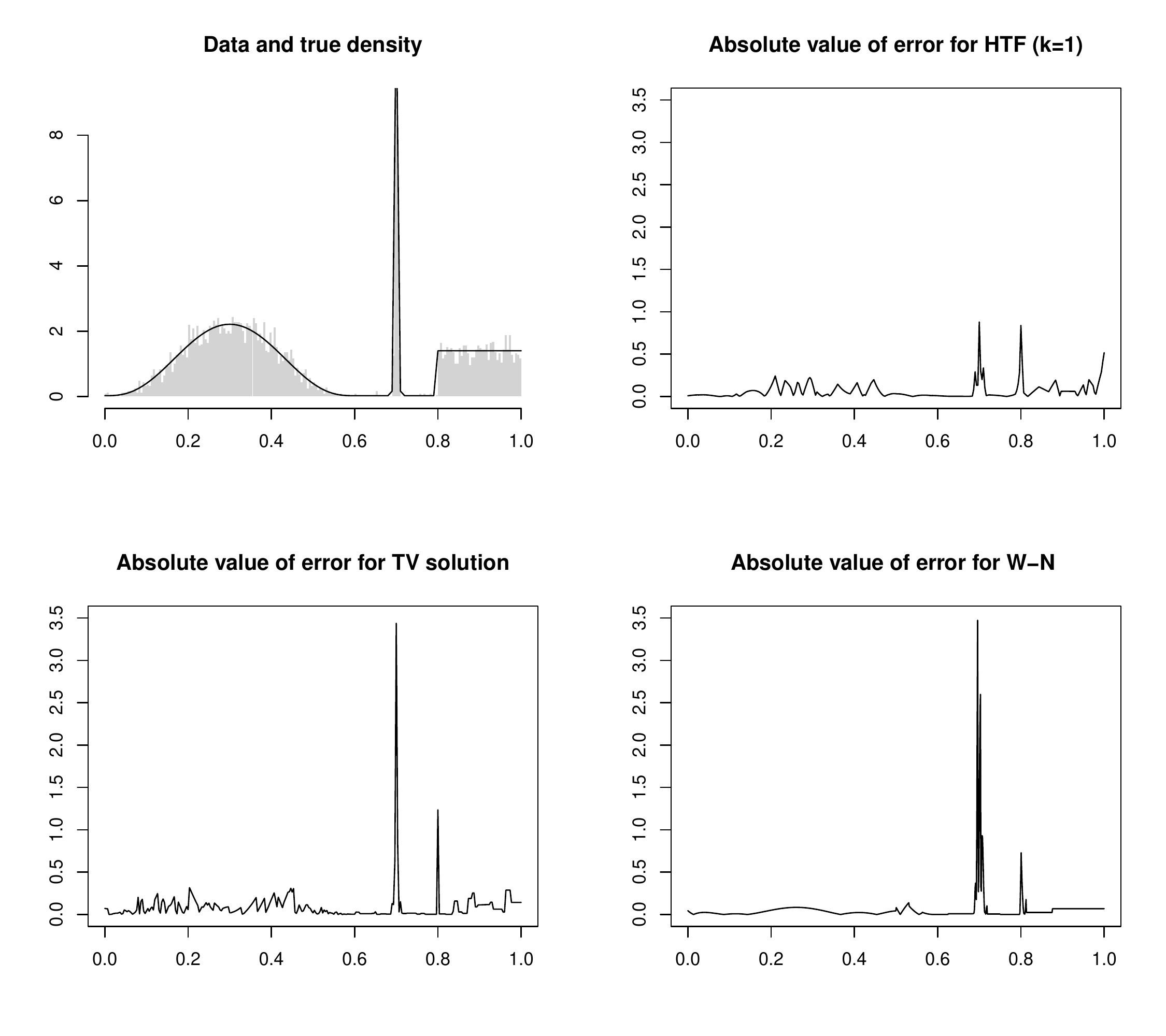} 
\caption{\label{fig:sim_examples3} The left panel at the top shows the data with the true density with sample size $n=5000$. The right panel at the top shows the plot of the vector obtained by taking the absolute value of the difference between HTF(k=1,using the full path) and the true density. The same is done for TV and W-N in the two panels at the bottom.  }
\end{center}
\end{figure}

\begin{table}[t]
\centering
\caption{\label{tab:sim3} Mean-squared error $\times$ 10 on example 3}
\medskip
\begin{small}
\begin{tabular}{p{3pc} p{5.5pc} p{5.5pc} p{5.5pc} p{5.5pc} p{5.5pc} p{5.5pc}}
n               & HTF ($k=1$) & HTF(k=1,no full path)  & TV        & Taut string & W-N     \\    
500             & 1.4         & 1.5                    & 30        & 22          & 3.7     \\   
  1000          & 1.0         & 1.1                    & 19        & 7.4         & 1.1     \\   
  2000          & 0.4         & 0.5                    & 11        & 3.2         & 0.5      \\  
  4000          & 0.2         & 0.2                    & 5.6       & 2.1         & 0.3     \\   
  5000          & 0.2         & 0.2                    & 4.9       & 2.1         & 0.3      \\  
\end{tabular}
\end{small}
\end{table}

\begin{table}[t]
\centering
\caption{\label{tab:sim4} Time in seconds on example 3 for the different methods}
\medskip
\begin{small}
\begin{tabular}{p{3pc} p{5.5pc} p{5.5pc} p{5.5pc} p{5.5pc} p{5.5pc} }
n               & HTF ($k=1$) & HTF(k=1,no full path)  & TV        & Taut string & W-N     \\   
500             & 1.03        & 0.02                   & 6. 85     & 0.02        & 1.15   \\   
  1000          & 1.15        & 0.03                   & 18.2      & 0.02        & 4.54   \\   
  2000          & 1.32        & 0.04                   & 45.3      & 0.03        & 22.0   \\   
  4000          & 1.45        & 0.06                   & 136       & 0.07        & 113    \\   
  5000          & 1.85        & 0.09                   & 237       & 0.10        & 202    \\   
\end{tabular}
\end{small}
\end{table}

On the other hand, in the light of the previous two examples, we now only focus on two different variants of histogram trend filtering with $k=1$. First, we compute the solution
path of problem (\ref{eqn:histTF_objective})   and  then we choose the tuning parameter with the surrogate AIC criterion  described  in Section \ref{model_section}. Secondly, we use the same criteria only on a grid of values. These values are

\[
 \{\lambda^{*}/100, \lambda^{*}/10, \lambda^{*}, \lambda^{*}\,10, \lambda^{*}\,100\},
\]
where  $\lambda^{*} = n\,\|\left(\Delta^{(k+1)}\right)^{-}\|_1\,D^{-1}_n $.  Our motivation here comes from the statement in Theorem \ref{theorem_2}.


We use these two variants of our method by borrowing a density from \cite{willett2007multiscale} that consists of a mixture of beta distributions. Figure \ref{fig:sim_examples3}  illustrates a plot of this distribution. The explicit density is defined as

\[
  f_0(x)  = \frac{3}{5}\left(\beta_{[0,\frac{3}{5}]}(x;4,4) \right) + \frac{1}{10}\left( \beta_{[\frac{2}{5},1]}(x;4000,4000) \right) + \frac{1}{40}\left( \text{Unif}_{[0,1]}(x) \right) + \frac{11}{40}\left( \text{Unif}_{[\frac{4}{5},1]}(x) \right),
\]
where $\beta_{[a,b]}$  refers to a Beta distribution shifted and scaled to have support on the interval $[a, b]$ and integrate to one. We use this density to generate data for different sample sizes.

The results in \ref{tab:sim3} show that our methodology outperforms in accuracy the competitors. This is also visualized in Figure \ref{fig:sim_examples3}, where we can see that   W-N seems to provide better recovery that HTF in areas where the true density  behaves as smooth polynomials. However, HTF seems to be more reliable in areas where the true density changes drastically.

On the other hand, from Table \ref{tab:sim4}, it is clear that HTF is more efficient than W-N and TV which begin to have considerable problems to scale. Even computing the approximate solution path for HTF(k=1) seems hundreds of times faster than solving a single problem for other penalized method. 



\section{Conclusion}
In summary, we have shown that histogram trend filtering can be successfully applied to the problem of density estimation. This estimator enjoys both computational and theoretical attractive  properties. On the computational side, our experiments suggests  that histogram trend filtering scales remarkably well with sample size, and that in practice it is just as computationally efficient as widely used methods based on kernel density estimation (KDE). However, unlike such methods, histogram trend filtering does not suffer from simultaneous over- and under-smoothing. Rather, our estimator can easily adapt to different levels of smoothness of the unknown true density.

Many methods have been proposed in the literature to deal with the problem of local adaptivity, e.g \citep{willett2007multiscale,sardy2010density}. As our paper has shown, these methods face challenges specifically in regions where the smoothness of true density changes rapidly. We have shown that histogram trend filtering can better adjust to such situations, while overcoming the scalability problems also inherent to other penalized methods. Thus histogram trend filtering enjoys both the computational efficiency of KDE methods and the adaptive properties of penalized estimators.  Finally,  our risk bounds provide strong theoretical guarantees of good performance for histogram trend filtering when seen as a variational problem or by its convex optimization formulation. This combination of practicality with strong statistical guarantees makes histogram trend filtering an ideal candidate for use in routine data-analysis applications that call for a quick, efficient, accurate density estimate.

\appendix

\section{Proof of technical results }

\subsection{ Proof of Theorem ~$\text{\ref{equivalent_formultion}}$ }

\begin{proof}

Let us assume that  $\hat{\theta}$ solves (\ref{eqn:histTF_objective}). Then, we define $\hat{g}_i =  \hat{\theta}_i - \log(n\,\delta_n) $  and $c =  \| \Delta^{(k+1)}\hat{\theta} \|_q^p $. Hence from the KKT conditions (\ref{eqn:histTF_objective}) is equivalent to
\[
 \begin{array}{ll}
    \underset{ \theta }{ \text{minimize} } &  \,\sum_{i= 1}^{D_n}\left\{\exp(\theta_i) - x_i\theta_i \right\}   \\
      \text{subject to} &   \| \Delta^{(k+1)}\theta \|_q^p \leq c. \\
 \end{array}
\]
Now, with the change of variable $g = \theta + \log(n\,\delta_n)$  and dividing by $n$ this is equivalent to

\begin{equation}\label{t1_step1}
 \begin{array}{ll}
    \underset{ g }{ \text{minimize} } & \,\delta_n \,\sum_{i= 1}^{D_n}\exp(g_i) - \frac{1}{n}\sum_{i= 1}^{D_n}\,x_i\,g_i    \\
      \text{subject to} &   \| \Delta^{(k+1)}g\|_q^p \leq c. \\
 \end{array}
\end{equation}
Next we define the function
\[
G(g) =  \delta_n \,\sum_{i= 1}^{D_n}\exp(g_i) - \frac{1}{n}\sum_{i= 1}^{D_n}\,x_i\,g_i.
 \]
 and for an arbitrary $g \in \mathbb{R}^{D_n}$ we define $g^{\prime} \in \mathbb{R}^{D_n}$ as

 \[
 g^{\prime}_i  = g_i - \log\left(\delta_n\,\sum_{j=1}^{D_n}\exp(g_j)\right).
 \]
 Then
 \[
    G(g^{\prime})   = G(g)  +  1 - \delta_n\,\sum_{j=1}^{D_n}\exp(g_j) +  \log\left(\delta_n\,\sum_{j=1}^{D_n}\exp(g_j)\right) \leq G(g)
 \]
since $t - \log(t) \geq 1$ for all $t > 0$. Moreover,
\[
    \| \Delta^{(k+1)}g\|_q^p =    \| \Delta^{(k+1)}g^{\prime}\|_q^p .
\]
Therefore, problem (\ref{t1_step1}) is equivalent to

\[
 \begin{array}{ll}
    \underset{ g }{ \text{minimize} } &  - \frac{1}{n}\sum_{i= 1}^{D_n}\,x_i\,g_i    \\
      \text{subject to} &  \,\delta_n \,\sum_{i= 1}^{D_n}\exp(g_i) = 1 \\
      &   \| \Delta^{(k+1)}g\|_q^p \leq c \\
   \end{array}
\]
and the claim follows.

\end{proof}



\subsection{ Proof of Theorem ~$\text{\ref{convergence_rates}}$ }

\begin{proof}
We now focus on the proof of Theorem~$\text{\ref{convergence_rates}}$. Since this requires several steps, we start by introducing some notation. Let $S$  be a set of integrable functions with support  $[a,b]$ and $d_S$  a metric on $L_1(\mathbb{R})$. For a given $\delta >0$, we define the entropy of $S$, denoted by $N\left(\delta,S, d_S \right) $, to be the minimum $N$ for which there exist integrable functions $f_1,\ldots\,f_N$  satisfying

\[
   \underset{f_i} {\text{min }}\,\,\, d_S(f_i,g) \leq  \delta, \,\,\,\forall g \in S.
 \]
The $\delta-$bracketing number  $N_{[\,]}\left(\delta,S,d_S\right) $ is defined as the minimum number of brackets of size $\delta$ required to cover $S$, where a bracket of size $\delta$ is a set of the form  $\left[l,u\right] := \{h: \,\ l(x)\leq h(x) \leq u(x)\,\, \forall x  \}$,  where $l$  and $u$  are non-negative integrable functions  and $d_S(l,u) < \delta$. 

From now on we denote by $d$ the distance

\[
d(g,h) =  \left[ \int_{\mathbb{R}} \left\{g(\mu)^{1/2}-h(\mu)^{1/2}\right\}^2 d\mu \right]^{1/2}.
\]
Moreover,  for an open set $\Omega \subset \mathbb{R}$  we denote by  $C^m(\bar{\Omega})$ the set of of $m$-times differentiable functions for which the derivatives of orders less than or equal to $m$ are uniformly continuous.

Next recall that the Sobolev space $W^{m,p}(\Omega)$ is endowed with the norm

\[
   \| u\|_{W^{m,p}(\Omega)} :=  \left( \sum_{j=0}^{m} \| u^{(j)}\|_{L^p(\Omega)}  \right)^{1/p}
\]
where $u^{(j)}$ is the $j-th$ weak deriuvative of $u$.

In  what follows we focus on the proof for the minimization over $S_{n,2}$. For the case $S_{n,1}$ we then briefly  describe the corresponding modification. The proof for $S_{n,i}^{\prime}$, $i = 1,2$  is analogous.



\begin{lemma}
\label{lemma_1}
There exists a constant $\delta_0 > 0$     such that if  $0 < \delta < \delta_0$, then for all $n$ we have

\[
\begin{array}{lll}
 \log\left(N(\delta,S_{n,2},\|\cdot\|_{\infty}) \right)  &\leq &  A\log\left(  \frac{(1+2\,e\,e^{T_n})\,T_n\,}{\delta}
   \right)\\
    \log\left(N(\delta,S_{n,1},\|\cdot\|_{\infty}) \right)  &\leq &  B\left(\frac{ T_n\,(2\,e\,e^{T_n}+1) }{\delta}  \right)^{1/(k+1)}  + (k+1)\log\left(\frac{\,T_n\,(1+2\,e\,e^{T_n})   }{\delta} \right)
    \end{array}
\]
where $A$  and $B$  are positive constants.

\end{lemma}


\begin{proof}

For $\epsilon >0 $ we first define the set
\[
\begin{array}{lll}
S(\epsilon)  & =  & \big\{ g \in C^{k+1}(\bar{\Omega}) \,:\,  \,\,\, \int_{[0,1]}\vert \left(g\right)^{(k+1)}(\mu) \vert^2 d\mu \leq  T_n   + \epsilon\, \\
 & &\,\,\,\,\,\,\,\,\,\,\,\,\,\,\,\,\,\,\,\,\,\,\,\,\,\, \,\,\quad\quad\quad \|g\|_{L^{\infty}([0,1])  }\leq  T_n  \,+\, \epsilon \,\big\},
\end{array}
\]

Then from Example 2.1 in \cite{van1990estimating}, because $T_n$ is bounded by below, there exists $\delta_0$ such that $\delta \in (0,\delta_0)$  implies
 \[
    \text{log}\left(N\left(\delta,S(\epsilon),\|\cdot\|_{\infty}\right)\right) \leq A\,\text{log}\left(\frac{T_n + \epsilon}{\delta}\right).
 \]

Next let $\epsilon > 0$ fixed.  Then  if   $f \in S_{n,2}$, by definition, there exists $h$ such that $\log(h) \in W^{k+1,2}(\Omega)$ and

\[
   \|f - h \|_{L^{\infty}(\Omega)} < \delta
\]
and
\[
\max\{\|\log(h)\|_{L^{\infty}(\Omega)},\|\log(h^{(k+1)})\|_{L^{2}(\Omega)}^2\} \leq T_n.
\]
Since $C^{k+1}(\bar{\Omega})$ is dense in $W^{k+1,2}(\Omega)$,  e.g \citep{adams2003sobolev,oden2012introduction},  by the Sobolev embedding theorem there exists  $g \in S(\epsilon)$ such that

\[
   \|g - \log(h) \|_{L^{\infty}(\Omega)} < \delta.
\]





 Let us now  set $N = N\left(\delta,S(\epsilon),\|\cdot\|_{\infty}\right)$ and let $g_1,\ldots,g_N \in S(\epsilon)$ be functions such that for every $g \in S(\epsilon)$, there exists $i \in \{1,\ldots,N\}$ satisfying

 $$\|g -g_i\|_{L^\infty([0,1])} \leq \delta. $$
Then for $f \in S_{n,2}$  choosing $h$ as before and  $g_i \in S(\epsilon)$ such that $\|g_i - \log(h) \|_{L^{\infty}(\Omega)} < 2\,\delta$ we obtain that for all $x \in (0,1)$

\[
  \begin{array}{lll}
    \vert f(x) - e^{g_i(x)}\vert &  \leq    &  \text{max}\left(h(x),e^{g_i(x)}\right)\vert \text{log}\left(h(x)\right) - g_i(x)\vert  + \delta\\
    & \leq  &  \, \text{max}\left(h(x),e^{\text{log}\left(h(x)\right) + 2\,\delta}\right)\,2\,\delta + \delta  \\
  & \leq  &  (1+\,2\, e\,e^{T_n})\,\delta.
  \end{array}
\]
Therefore, for all $\epsilon > 0$

\[
  \begin{array}{lll}
 \text{log}\left(N\left(\delta,S_{n,2},\|\cdot\|_{\infty}\right)\right)   &  \leq &   \text{log}\left(N\left(\frac{\delta}{(1+2\,e\,e^{T_n})},S(\epsilon),\|\cdot\|_{\infty}\right)\right)\\
  &\leq  &  A\,\text{log}\left( (1+2\,e\,e^{T_n})\frac{ T_n + \epsilon}{\delta}\right).
   \end{array}
\]
Letting epsilon go to zero we arrive to

\[
 \text{log}\left(N\left(\delta,S_{n,2},\|\cdot\|_{\infty}\right)\right)     \leq  A\,\text{log}\left( \frac{T_n\,(1+2\,e\,e^{T_n}) }{\delta}\right)
\]
Hence the result follows for $S_{n,2}$. The proof for the sieve  $S_{n,1}$ follows the same lines with entropy bound for the corresponding $S$  coming from the proof of Theorem 2 in \cite{mammen1991nonparametric}.

\end{proof}

\begin{corollary}
With the notation from the previous lemma,  there exists $\delta_{0}^{\prime}$  such that   $0 < \delta < \delta_{0}^{\prime}$ implies
\[
\begin{array}{lll}
 \log\left\{N_{[\,]}\left(\delta,S_{n,2},d \right)\right\} & \leq & A\,\log\left( \frac{2\, T_n\,(e^{T_n}\,2\,e+1) }{\delta^{2}} \right),\\
   \log\left\{N_{[\,]}\left(\delta,S_{n,1},d \right)\right\} & \leq &  (k+1)\log\left(  \frac{2\,T_n\,(e^{T_n}\,2\,e+1) }{\delta^{2}} \right)\\
   & & +  B\left( \frac{2\, T_n\,(e^{T_n}\,2\,e+1) }{\delta^{2}}\right)^{1/(k+1)}.
\end{array}
\]
\end{corollary}

\begin{proof}
Let $j \in \{1,2\}$. We proceed as in the proof of Lemma 3.1 from \cite{ghosal2001entropies}. First,  given $\delta \in (0,\delta_0)$, we  define $\eta = (2)^{-1}\delta^2 $. Next, let $f_1,\ldots,f_N$ be  non-negative integrable functions  with support in $[0,1] $  such that  for all $h \in S_{n,j}$, there exists $i \in \{1,\ldots,N\}$ such that $\|f_i - h\|_{L^\infty(\Omega)} < \eta$. We then construct the brackets $[l_i,u_i] $ by defining

\[
    l_i = \text{max}\left(f_i - \eta, 0\right),\,\,\, u_i = \left(f_i + \eta \right) 1_{[0,1]}.
\]
Then  $ S_{n,j}  \subset  \cup_{i=1}^{N} [l_i,u_i]  $. Since  $0 \leq u_i - l_i  \leq 2\eta$, we obtain

\[
    \int_{-\infty}^{\infty}\left\vert u_i(\mu) - l_i(\mu)\right\vert d\mu  =  \int_{0}^{1}\left\vert u_i(\mu) - l_i(\mu)\right\vert d\mu \leq   2\,\eta.
\]
 Therefore,

\[
    N_{[\,]}\left(2\,\eta,S_{n,j}, \|\cdot\|_1\right)  \leq N.
\]
The results then follows from the previous lemma  by choosing $\eta = (2\,)^{-1}\,\delta^2$, implying that

\[
\begin{array}{lll}
  \log\left\{N_{[\,]}\left(\delta,S_{n,j},d \right)\right\}  & \leq &   \log\left\{N_{[\,]}\left(\delta^{2},S_{n,j},\|\cdot\|_1 \right)\right\} \\
   &\leq &   \log\left\{N\left(\frac{\delta^{2}}{2\,},S_{n,j},\|\cdot\|_{\infty} \right)\right\}.
\end{array}
\]
\end{proof}

\paragraph{Existence}

We now show that the sets $M_{n,i}$ and $M_{n,i}^{\prime}$ are not empty. To this end, note that in $\left(\bar{C}(\Omega ), \|\cdot\|_{L^{\infty}(\Omega )} \right)$ we have that

$$\{h\,:\, e^h \in S_{n,i} \}  \subset \text{cl}_{\Omega }\left(\mathcal{P} \cap \bar{C}(\Omega) \right) \cap   \text{cl}_{\Omega}\left(h \,:\, h \in \ W^{k+1,i}(\Omega), \,\,h\,\,\text{ is continuous}   \right)  $$

hence by the Sobolev embedding theorem we obtain that $\{h\,:\, e^h \in S_{n,i} \}$ is compact in $\left(\bar{C}(\Omega), \|\cdot\|_{L^{\infty}(\Omega)} \right)$. Similarly, $\{h\,:\, e^h \in S_{n,i}^{\prime} \}$ is also compact in $\left(\bar{C}(\Omega), \|\cdot\|_{L^{\infty}(\Omega)} \right)$.

\paragraph{Rates}
We conclude the proof by using  Theorem 1  from \cite{wong1995probability}. First, we observe that if $\alpha \in (0,1]$, then $\epsilon_n =  \left(\log n\right)^{q}n^{-\alpha/2}$ satisfies

\[
  \begin{array}{lll}
  \int_{  \frac{ \epsilon_n^{2}}{2^{8}}  }^{ \sqrt{2}\epsilon_n }   \sqrt{ \log\left(N_{[]}\left( u/c_3,S_{n,2},d\right) \right) }du &\leq &   2^{1/2}A^{1/2}\,\epsilon_n\,\left(\log\left( \frac{2\,T_n\,(2\,e^{T_n}\,e+1)\,c_{3}^{2}2^{16} }{\epsilon_n^{4}} \right)  \right)^{1/2}\\
    &\leq  & c_4\,n^{1/2}\epsilon_n^{2}   \\
  \end{array}
\]
for large enough $n$ where  $B$ is some positive constant and $c_3$ and $c_4$ are given as in Theorem 1 from \cite{wong1995probability}.  Hence, the claims follows for $S_{n,2}$.

 To conclude the proof for the sieve $S_{n,2}^{\prime}$, we observe that for any observation $y_j$  there exists $\xi_{j^{\prime}}$  such that $y_j$ and  $\xi_{j^{\prime}}$  belong to the same bin. With an abuse of notation we will denote such  $\xi_{j^{\prime}}$ as $\xi_j$. Then for positive constant $c$, if $n$ is large enough we have

\[
\begin{array}{lll}
  \mathrm{P}^*\left\{ \underset{\hat{f}_n \in M_{n,2}^{\prime}}{\text{sup }} d\left(\hat{f}_n,f_0\right) \geq \epsilon_n \right\}  & \leq & \underset{g\in S_{n,2}^{\prime}}{\inf }  \mathrm{P}^*\left\{ \underset{  h \in S_{n,2}^{\prime}\,:\, d\left(h,f_0 \right) \geq \epsilon_n }{\text{sup} } \prod_{j=1}^{n} h(\xi_j)/g(\xi_j) \geq \text{exp}\left(-c\,n\,\epsilon_n^{2} \right)    \right\}. \\
\end{array}
\]
But for any $g,h \in S_{n,2}^{\prime}$  we have

\[
    \prod_{i=1}^{n} h(\xi_j)/g(\xi_j)  =  \prod_{j=1}^{n} \left\{ h(y_j)/g(y_j) \right\}  \left\{ g(y_j)/g(\xi_j) \right\} \left\{ h(\xi_j)/h(y_j) \right\},
\]
and by the Lipschitz continuity condition,

\[
\begin{array}{lll}
  \prod_{i=1}^{n}  h(\xi_j)/h(y_j)   &\leq &   \left\{ 1  +  \frac{\,r\,h(y_j)^{-1}}{D_n}  \right\}^{n} \\
&  \leq &  \left( 1  +  \frac{r\,}{\,n^{\alpha}}  \right)^{n}.
  \end{array}
\]
Similarly,

is not empty,\[
\prod_{i=1}^{n}  g(y_j)/g(\xi_j) \leq   \left( 1  +  \frac{r\,}{\,n^{\alpha}}  \right)^{n}.
\]
Then for large enough $n$,

\[
\begin{array}{lll}
  \mathrm{P}^*\left\{ \underset{\hat{f}_n \in M_{n,2}^{\prime}}{\text{sup }} d\left(\hat{f}_n,f_0\right) \geq \epsilon_n \right\}  & \leq & \underset{g\in S_{n,2}^{\prime}}{\inf } \mathrm{P}^*\Bigg[ \underset{  h \in S_{n,2}^{\prime}\,:\, d\left(h,f_0 \right) \geq \epsilon_n }{\text{sup} } \prod_{j=1}^{n} h(y_j)/g(y_j) \geq \,\text{exp}\big\{-c\,n\,\epsilon_n^{2} \\
  & &  \,\,\,\,\,\,\,\, - 2\,n\text{log}\left( 1 + r/n^{\alpha}\right)  \big\}    \Bigg] \\
  & \leq &  \underset{g\in S_{n,2}^{\prime}}{\inf } \mathrm{P}^*\left\{ \underset{  h \in S_{n,2}^{\prime}\,:\, d\left(h,f_0 \right) \geq \epsilon_n }{\text{sup} } \prod_{j=1}^{n} h(y_j)/g(y_j) \geq \,\text{exp}\left(-c_1\,n\,\epsilon_n^{2}\right)    \right\} \\
  & \leq &   6\,\text{exp}\left(-c_2\,n\,\epsilon_{n}^2\right),
\end{array}
\]
if  $0 < c  < c_1$ where $c_1$ and $c_2$ can be obtained from Theorem 3 from \cite{wong1995probability} and $\text{P}^*$ is understood as the outer measure corresponding to $f_0$.




Finally,  we replace $S_{n,2}$ by $S_{n,1}$, and set $\epsilon = n^{-t}$ as in the statement of the theorem. Then  the same argument from above shows that the solution set $M_{n,1}$ is not empty. A similar argument as above leads to the desired conclusion for both sieves  $S_{n,1}$  and $S_{n,1}^{\prime}$




\end{proof}

\subsection{Proof of Theorem \ref{theorem_1}}

Throughout we define the vectors
\[
  \begin{array}{l}
    g_0(\xi)  = \left\{\text{log}f_0(\xi_1),\ldots,\text{log}f_0(\xi_{D_n}) \right\}\\
    \hat{g}(\xi)  = \left\{\text{log}\hat{f}(\xi_1),\ldots,\text{log}\hat{f}(\xi_{D_n}) \right\}\\
    f_0(\xi)^{1/2} =  \left\{ f_0(\xi_1)^{1/2},\ldots,f_0(\xi_{D_n})^{1/2}   \right\}.
   \end{array}
\]




\begin{proof}
We first prove the case  where $p=q=1$. To that end let $k_1 > 0  $ and  $m_1 > 0$  be a lower bound and upper bound on the true density. Let us also define

\[
   \hat{\theta} = \underset{\theta}{\arg \min  } \,\,\left\{ l(\theta) +  \tau_n \|\Delta^{(k+1)}\theta \|_q^p \right\}.
\]
We consider the function $G : \mathbb{R}^{D_n}  \rightarrow \mathbb{R}$  as
\[
 G(u) =   l(\theta^0 + u) - l(\theta^0) + \tau_n \left(  \|\Delta^{(k+1)}u + \Delta^{(k+1)}\theta^0 \|_1  -  \|\Delta^{(k+1)}\theta^0\|_1  \right),
\]
and we show that with high-probability this function is strictly positive in the boundary of the Euclidean unit ball.  The result will then follow because $G$ is convex,  $G(0) = 0$, and $G(\hat{\theta} - \theta ) \leq 0$.

To show this we first observe that by the taylor's expansion, there exists  $\alpha_j \in [0,1]$  such that

\[
\begin{array}{lll}
      l(\theta^0 + u) - l(\theta^0)  & =  & \mathlarger\sum\limits_{j=1}^{D_n} \left\{ \text{exp}(\theta_j^0 + u_j)  -  x_j\,(\theta_j^0 + u_j)

      +  x_j\,\theta_j^0 -   \text{exp}(\theta_j^0 )  \right\}  \\

      & =  &  \mathlarger\sum\limits_{j=1}^{D_n} \left\{ u_j \text{exp}(\theta_j^0)-u_j x_j + 2^{-1} \text{exp}(\theta_j^0 +  \alpha_j u_j)  u_j^2 \right\} \\
      &\geq  & -\left\| \text{exp}(\theta^0)-x \right\|_{\infty}  D_n^{1/2} \|u\|_2    +  2^{-1}n\,D_n^{-1}\,k_1\,e^{-\|u\|_2} \|u\|_2^2.
\end{array}
\]
On the other hand,
\[
\begin{array}{lll}
  \tau_n \left(  \|\Delta^{(k+1)}u + \Delta^{(k+1)} \theta^0\|_1  -  \|\Delta^{(k+1)}\theta^0\|_1  \right) &\geq & -\tau_n\,\|\Delta^{(k+1)}u\|_1  \\
  &  \geq  &     -\tau_n\,k_4\,\|\Delta^{(1)}u\|_1 \\
  & \geq  & -\tau_n\,k_4\,D_n^{1/2}\,\|\Delta^{(1)}u\|_2 \\
  & \geq  & -\tau_n\,D_n^{1/2}\,k_2\, \|u\|_2
\end{array}
\]
for positive constants $k_2,k_3$ and $k_4$. Therefore, if $\|u\|_2 =  1$  we obtain

\[
   G(u)  \, \geq  \, - \left\| \text{exp}(\theta^0)-x \right\|_{\infty} D_n^{1/2} + 2^{-1}n\,D_n^{-1}\,k_1\,e^{-1}  -\tau_n\,D_n^{1/2}\,k_2\,  \,\,>\,\,0
\]
if only if

\[
     n\,\,k_1\,e^{-1}   >  2\,D_n^{3/2}\left\{ \left\| \text{exp}(\theta^0)-x \right\|_{\infty}  +  k_2\tau_n \right\}.
\]
Next, we observe that by the mean value theorem for integrals we have  $x_i \sim  \text{Binomial}\{D_n^{-1}f_0(z_i) ,n\}$  for some $z_i$ in bin $i$. Then for any $t > 0$  and $i \in \{1,\ldots,D_n\}$,  using Chernoff's bound we obtain

\[
\begin{array}{lll}
  \mathrm{P}\left\{ \vert x_i - n\,D_n^{-1}\,f_0(\xi_i) \vert \geq  t\right\}  & \leq  &   \mathrm{P}\left\{ \vert x_i - n\,D_n^{-1}\,f_0(z_i) \vert \geq  t   -  \left \vert n\,D_n^{-1}\,f_0(\xi_i) -  n\,D_n^{-1}\,f_0(z_i) \right\vert   \right\} \\
  & \leq &   \exp\{-\epsilon^2(2 + \epsilon)^{-1} n\,D_n^{-1}\,f_0(z_i)\} + \exp\{-\epsilon^2\,2^{-1} n\,D_n^{-1}\,f_0(z_i)\}\text{,}
  \end{array}
\]
if  $\epsilon > 0 $ where
\[
\begin{array}{lll}
\epsilon & = &   t\,n^{-1}\,D_n \,f_0(z_i)^{-1}  -  f_0(z_i)^{-1} \,\left \vert f_0(\xi_i) -  f_0(z_i) \right\vert.    \\
  \end{array}
\]
We choose $t = C\,n\,D_n^{-3/2}$  for some constant $C  >0$  and observe that the respective  $\epsilon $ is positive for large enough $n$. To see this  we observe that

\[
\begin{array}{lll}
\epsilon & = &  f_0(z_i)^{-1}C\,D_n^{-1/2}   -  f_0(z_i)^{-1} \,\vert f_0(\xi_i) - f_0(z_i) \vert    \\
&  \geq &   f_0(z_i)^{-1}\left( C\,D_n^{-1/2}   -   L\,D_n^{-1} \right).
  \end{array}
\]
 Hence for large enough $n$  we see that
\[
\begin{array}{lll}
\mathrm{pr}\left( \left\| \text{exp}(\theta^0)-x \right\|_{\infty}\geq  n\,D_n^{-3/2}\,C\right)  & \leq & \mathlarger\sum\limits_{j=1}^{D_n} \mathrm{ pr}\left\{   \vert x_j - n\,D_n^{-1}\,f_0(\xi_j) \vert \geq  n\,D_n^{-3/2}\,C \right\} \\
& \leq   & \mathlarger\sum\limits_{j=1}^{D_n}\, \text{exp} \left[  -\frac{\left\{\frac{n\,D_n^{-3/2}\,C}{n\,D_n^{-1}\,f_0(z_i)}  -  \left \vert \frac{f_0(\xi_i) }{  f_0(z_i)  } -  1 \right\vert\right\}^2 \, n\,D_n^{-1}\,\,f_0(z_i)   }{2 + \frac{n\,D_n^{-3/2}\,C}{n\,D_n^{-1}\,f_0(z_i)}  -  \left \vert \frac{f_0(\xi_i) }{f_0(z_i)} -  1 \right\vert }  \right]   \\
& &   + \mathlarger\sum\limits_{j=1}^{D_n}  \, \text{exp} \left[  -\left\{ \frac{n\,D_n^{-3/2}\,C}{n\,D_n^{-1}\,f_0(z_i)}  -  \left \vert \frac{f_0(\xi_i)}{f_0(z_i)} -  1 \right\vert\right\}^2\, n\,2^{-1}\,D_n^{-1}\,f_0(z_i)   \right]    \\

&\leq  &  \mathlarger\sum\limits_{j=1}^{D_n} \text{exp} \left\{  -\frac{   f_0(z_i)^{-2} \left( C\,D_n^{-1/2}   -   L\,D_n^{-1} \right)^2 \,n\,D_n^{-1}\,f_0(z_i)}{2 + \frac{\,D_n^{-1/2}\,C  }{\,f_0(z_i)}  -  \left \vert \frac{f(\xi_i)}{f_0(z_i)} -  1 \right\vert }    \right\}    \\
&  &  + \mathlarger\sum\limits_{j=1}^{D_n}   \text{exp} \left\{  -2^{-1}f_0(z_i)^{-2} \left( C\,D_n^{-1/2}   -   L\,D_n^{-1} \right)^2 \,n\,D_n^{-1}\,f_0(z_i)    \right\}.     \\
 \end{array}
\]
Therefore,  if $D_n = a_n\,n^{1/s}$  as in the statement of the theorem, then

\[
\begin{array}{lll}
\mathrm{pr} \left\{ \left\| \text{exp}(\theta^0)-x \right\|_{\infty} \geq  n\,D_n^{-3/2}\,C\right\}  & \leq & \mathlarger\sum\limits_{j=1}^{D_n} \text{exp} \left\{  -\frac{f_0(z_i)^{-2} \left( C   -   L\,D_n^{-1/2} \right)^2 \,n\,D_n^{-2}\,f_0(z_i)   }{2 + \frac{\,D_n^{-1/2}\,C}{\,f_0(z_i)}   }   \right\} \\
& & + \mathlarger\sum\limits_{j=1}^{D_n}   \text{exp} \left\{  - f_0(z_i)^{-2} \left( C   -   L\,D_n^{-1/2} \right)^2 n\,2^{-2}\,D_n^{-2}\,f_0(z_i) \right\}  . \\
\end{array}  \\
\]
Hence we set $C  = 4^{-1}\,c\,k_1 \,e^{-1}   $ for some $c \in (0,1)$,  choosing $r =  4^{-1}\,\left(1-c\right)k_1 \,e^{-1}\,k_2^{-1} \,$ ensures that with high probability,

\[
     n\,\,k_1\,e^{-1}   >  2\,D_n^{3/2}\left\{ \| \text{exp}(\theta_j^0)-x_j \|_{\infty}  +  k_2\tau_n \right\}.
\]

If, on the other hand, $p=q=2$, then the proof follows the same lines, with the main modification involving the following bound:
\[
\begin{array}{lll}
  \tau_n \left\{  \|\Delta^{(k+1)}u + \Delta^{(k+1)} \theta^0\|_2^2  -  \|\Delta^{(k+1)}\theta^0\|_2^2  \right\} &\geq & \tau_n\left( \,\|\Delta^{(k+1)}u\|_2^2  -  2\,\|\Delta^{(k+1)}u\|_2\,\|\Delta^{(k+1)} \theta^0\|_2  \right) \\
  & \geq  & -\tau_n\,2\,\|\Delta^{(k+1)} g(\xi)\|_2\, \|\Delta^{(k+1)}u\|_2  \\
  & \geq & -\tau_n\,k_5\, \| g(\xi)\|_2\,\|\Delta^{(k+1)}u\|_2 \\
  &  \geq &  -\tau_n\,k_6\, D_n^{1/2} \|u\|_2.
\end{array}
\]
for some positive constants $k_5 $ and $k_6$.
\end{proof}

\subsection{Proof of Theorem \ref{theorem_2}}




Before beginning the proof of the claim we start by proving an auxiliary lemma.

\begin{lemma}
\label{auxiliar_lemma}
With the notation from Theorem \ref{theorem_2}, if $a \in \mathbb{R}^{D_n}$, then
\[
\text{P}\left( \vert\left( x - \exp(\theta^{0}) \right)^{T}a\vert \geq \frac{n\,\|a\|_{\infty}}{D_n^{r}} \right)  \leq 4\exp\left(-c_r\frac{n}{D_n^{2r}}\right) 
\]
for all $r>0$ and some positive constant $c_r$ depending on $r$.
\end{lemma}

\begin{proof}
Our proof is inspired by the construction in Lemma 3 from \cite{devroye1983equivalence}. We start by denoting $p_i = \frac{\exp\left(\theta_i^{0}\right)}{n}$, $i =1,\ldots, D_n$. Then we can think of $x_i$ as the occurrences of value $i$ among $u_1,\ldots,u_n$ where $\text{P}(u_k =  j) = p_j$ for $j=1,...D_n$ and $k = 1,2,\ldots$. Next, we define $N \sim \text{Poisson}\left(n\right)$, and $x_{i}^{\prime}$ as the occurrences of value $i$ among $u_1,\ldots,u_N$. Clearly,  $x_{i}^{\prime} \sim \text{Poisson}\left(n\,p_i\right)$. Moreover,

\[
  \left\vert \sum_{i=1}^{D_n}a_i\left(x_i - n\,p_i\right)\right\vert \leq  \left\vert \sum_{i=1}^{D_n}a_i\left(x_i^{\prime} - n\,p_i\right)\right\vert +  \left\vert \sum_{i=1}^{D_n}a_i\left(x_i - x_i^{\prime}\right)\right\vert,
\]
form which
\begin{equation}
\label{lemma_1}
  \text{P}\left( \left\vert \sum_{i=1}^{D_n}a_i\left(x_i - n\,p_i\right)\right\vert \geq 2\epsilon  \right) \leq \text{P}\left(\|a\|_{\infty}\vert N-n \vert \geq \epsilon \right) + \text{P}\left( \left\vert \sum_{i=1}^{D_n}a_i\left(x_i^{\prime} - n\,p_i \right)\right\vert \geq \epsilon  \right)
\end{equation}
for all $\epsilon >0$. We now bound both terms in (\ref{lemma_1}).  First, we proceed using Hoeffding's inequality,

\[
\begin{array}{lll}
 \text{P}\left(  \sum_{i=1}^{D_n}a_i\left(x_i^{\prime} -n\,p_i\right)  \geq \epsilon  \right)& \leq &   \underset{ t > 0}{ \inf }\, \exp\left( -\epsilon\,t    + \sum_{i=1}^{D_n} n\,p_i\left( \exp(t\,a_i)-1-t\,a_i  \right)    \right) \\
  & \leq &   \underset{ t > 0}{ \inf }\, \exp\left( -\epsilon\,t    + n\,\left( \exp(t\,\|a\|_{\infty})-1-t\,\|a\|_{\infty}  \right)    \right)\\
   & \leq & \exp\left( -\frac{\epsilon}{D_n^{r}\,\|a\|_{\infty}} + n\left( \exp(\frac{1}{D_n^{r}}) -1 - \frac{1}{D_n^{r}}  \right) \right) \\
    &\leq & \exp\left( -\frac{\epsilon}{D_n^{r}\,\|a\|_{\infty}} + \frac{c\,n}{D_n^{2r}}   \right)
\end{array}
\]
for  some positive constant $c$ if $D_n^r$ is large enough. Therefore,  setting $\epsilon =  c_1\,n\,\|a\|_{\infty}\,D_n^{-r}$ with $c_1 > c$, we obtain
\[
 \text{P}\left(  \sum_{i=1}^{D_n}a_i\left(x_i^{\prime} -n\,p_i\right)  \geq c_1 \frac{n\,\|a\|_{\infty}}{D_n^{r}}  \right) \leq \exp\left( -\frac{(c_1-c)n}{D^{2r}} \right).
\]
 With union bound inequality and repeating the same argument from above,  we arrive to

\[
 \text{P}\left(  \vert \sum_{i=1}^{D_n}a_i\left(x_i^{\prime} -n\,p_i\right) \vert  \geq c_1 \frac{n\,\|a\|_{\infty}}{D_n^{r}}  \right) \leq 2\exp\left( -\frac{(c_2-c)n}{D^{2r}} \right).
\]
Finally, from the proof of Lemma 3 in \cite{devroye1983equivalence} we have

\[
    \text{P}\left(\|a\|_{\infty}\vert N-n \vert \geq  c_1 \frac{n\,\|a\|_{\infty}}{D_n^{r}} \right)  \leq  2\,\exp\left(  -\frac{c_1^{2}}{4}\frac{n}{D_n^{2r}}\right )
\]
and the result follows.

\end{proof}

\begin{proof}

 Let $e_1$  an element of the canonical basis in $\mathbb{R}^{D_n}$ and let us denote by $P$ the orthogonal projection onto the row space of $\Delta^{k+1}$.  We start by noticing that from sub-optimality we have

\[
    l(\hat{\theta}) +  \lambda\,\|\Delta^{(k+1)}\hat{\theta}\|_1 \leq l(\theta^{0}) +  \lambda\,\|\Delta^{(k+1)}\theta^{0}\|_1.
\]
Hence, setting $\lambda = \tau/2$, we obtain

\begin{equation}
\label{kl_1}
\begin{array}{lll}
   \sum_{j=1}^{D_n} \delta_n\,  f_0(\xi_j^{\prime}) \log\left(\frac{f_0(\xi_j^{\prime})}{\hat{f}(\xi_j^{\prime})}  \right)& \leq &   \frac{1}{n}\left(x - \exp(\theta^{0}) \right)^{T}\left( \left(\Delta^{(k+1)}\right)^{-}\Delta^{(k+1)}  + P_{R^{\perp}} \right)\left(\hat{\theta} - \theta^{0}\right)\\
   & & +  \frac{\lambda}{n}\left(\|\Delta^{(k+1)}\theta^{0}\|_1    - \|\Delta^{(k+1)}\hat{\theta}\|_1    \right).
\end{array}
\end{equation}
Next we bound each of the terms on the right hand side of (\ref{kl_1}). First, define $v_1,\ldots,v_{k+1}$  to be an orthonormal basis of $R^{\perp}$   such that $v_1 =  D_n^{-1/2}\left(1,\ldots,1\right)$. Then, it is not difficult to see that these vectors can be chosen to satisfy $\|v_j\|_{\infty} = O(D_n^{-1/2})$ for $j=1,\ldots,k+1$. Therefore, by Holder's inequality

\begin{equation}
\label{kl_2}
\begin{array}{lll}
 \frac{1}{n}\left(x - \exp(\theta^{0}) \right)^T\,P_{R^{\perp}}\left(\hat{\theta}-\theta^{0}\right) & = &  \frac{1}{n}\sum_{j=1}^{k+1} \left[ \left(x -\exp(\theta^0) \right)^T\,v_j  \right]   \left[ v_j^{T}\left( \hat{\theta} - \theta^{0}  \right) \right] \\
  & \leq &   \frac{c}{n}\,\|x -\exp(\theta^0)\|_1\,D_n^{-1/2}\,\left(\| \log(f_0(\xi^{\prime}))  \|_{\infty} + \|\log(\hat{f}(\xi^{\prime}))\|_{\infty}   \right ).
\end{array}
\end{equation}
It follows form Lemma 3 in \cite{devroye1983equivalence}  that

\[
   \frac{1}{n}\left(x - \exp(\theta^{0}) \right)^T\,P_{R^{\perp}}\left(\hat{\theta}-\theta^{0}\right)   = O_{\text{P}}\left(  \frac{1}{n^{1/2-b}}  \right).
\]
assuming that we constraint $\|\hat{\theta} - \log(n\,\delta_n)\|_{\infty} \leq n^{b}$.

On the other hand,

\begin{equation}
\label{kl_3}
\begin{array}{l}
\frac{1}{n}\left(x - \exp(\theta^{0}) \right)^{T}\left( \left(\Delta^{(k+1)}\right)^{-}\Delta^{(k+1)}  \right)\left(\hat{\theta} - \theta^{0}\right) \leq \\
  \frac{1}{n}\,\|\left(x - \exp(\theta^{0})\right)^{T}\left(\Delta^{(k+1)}\right)^{-}\|_{\infty}\left( \|\Delta^{(k+1)}\theta^{0}\|_1 + \|\Delta^{(k+1)}\hat{\theta} \|_1  \right).\\
\end{array}
\end{equation}
Moreover, from the previous lemma  we obtain
\[
\text{P}\left( \left\|\left( x - \exp(\theta^{0}) \right)^{T}\left(\Delta^{(k+1)}\right)^{-}  \right\|_{\infty} \geq \frac{n\,\|\left(\Delta^{(k+1)}\right)^{-}\|_{\infty}}{D^{r} } \right)  \leq 4\exp\left(-c_1\frac{n}{D_n^{2r} } + \log(D_n)\right)
\]

Therefore, combining $(\ref{kl_1})$, $(\ref{kl_2})$ and $(\ref{kl_3})$,   if $\lambda \geq \,\|\left(x - \exp(\theta^{0})\right)^{T}\left(\Delta^{(k+1)}\right)^{-}\|_{\infty}$, then,

\begin{equation}
\begin{array}{lll}
 \sum_{j=1}^{D} \delta\,  f(z_j) \log\left(\frac{f(z_j)}{\hat{f}(z_j)}  \right) & \leq & O_{\text{P}}\left( \frac{ \|\left(\Delta^{(k+1)}\right)^{-}\|_{\infty}  }{D_n^{r}}\|\Delta^{(k+1)}\theta^{0}\|_1 + \frac{1}{n^{1/2-b}} \right)
\end{array}
\end{equation}

\end{proof}

\begin{small}
\singlespacing
\bibliographystyle{abbrvnat}
\bibliography{deconvolution}

\end{small}

\end{spacing}
\end{document}